\documentclass[dvips,preprint]{imsart}

\RequirePackage[OT1]{fontenc}
\RequirePackage{amssymb,amsthm,amsmath}
\RequirePackage[numbers]{natbib}
\RequirePackage[colorlinks,citecolor=blue,urlcolor=blue]{hyperref}
\RequirePackage{graphicx}

\startlocaldefs
\numberwithin{equation}{section}
\theoremstyle{plain}
\newtheorem{defi}{Definition}[section]
\newtheorem{theo}{Theorem}[section]
\newtheorem{prop}{Proposition}[section]
\newtheorem{lemm}{Lemma}[section]
\newtheorem{rem}{Remark}[section]
\def\RR{\mathbb{R}}
\def\TT{\mathbb{T}}
\def\PP{\mathbb{P}}
\def\NN{\mathbb{N}}

\def\EE{\mathbb{E}}
\def\II{\mathbb{I}}
\def\dd{\mathrm{d}}
\def\ee{\mathrm{e}}
\def\ii{\mathrm{i}}

\endlocaldefs

\begin{document}

\begin{frontmatter}
\title{Exponential growth of bifurcating processes with
  ancestral dependence}
\runtitle{Bifurcating processes with ancestral dependence}

\begin{aug}
\author{\fnms{Sana} \snm{Louhichi}\thanksref{t1}
\ead[label=e1]{Sana.Louhichi@imag.fr}}
\and
\author{\fnms{Bernard} \snm{Ycart}\thanksref{t1}
\ead[label=e2]{Bernard.Ycart@imag.fr}}

\thankstext{t1}{Research supported by Laboratoire d'Excellence TOUCAN
  (Toulouse Cancer)}
\runauthor{S. Louhichi and B. Ycart}

\affiliation{Univ. Grenoble-Alpes and CNRS}

\address{Laboratoire Jean Kuntzmann\\
51 rue des Math\'ematiques 38041 Grenoble cedex 9, France}

\end{aug}

\begin{abstract}
Branching processes are classical growth models in cell kinetics. In
their construction, it is usually assumed that cell lifetimes are
independent random variables, which has been proved false in
experiments. Models of dependent lifetimes are considered here, in
particular bifurcating Markov chains. Under hypotheses of stationarity
and multiplicative ergodicity, the corresponding branching process is
proved to have the same type of asymptotics as its classic counterpart
in the i.i.d. supercritical case: the cell population grows
exponentially, the growth rate being related to the exponent of
multiplicative ergodicity, in a similar way as to the Laplace
transform of lifetimes in the i.i.d. case. An identifiable
model for which the
multiplicative ergodicity coefficients and the growth rate can be
explicitly computed is proposed.
\end{abstract}



\end{frontmatter}

\section{Introduction}
Let $\TT$ denote the infinite complete binary tree where each vertex
has exactly $2$ descendants. Let
$(T_v)_{v\in\TT}$ be a bifurcating process, i.e. a family of positive
random variables indexed by $\TT$, defined on a probability space
$(\Omega,\mathcal{F},\PP)$. The vertices
of $\TT$ are interpreted as cells, and $T_v$ as
the lifetime of cell $v$. The root (ancestor) of the tree being
born at time $0$, let $N_t$
the number of individuals alive at time $t$:
$(N_t)_{t\geqslant 0}$ is a continuous time branching process
(precise definitions
will be given in section \ref{PSBP}). If lifetimes are i.i.d.,
the population $N_t$ grows exponentially in $t$: this is
a particular case of one of the most basic results of the theory (see
Bellmann \& Harris \cite{BellmanHarris52}, Harris
\cite[Chap. VI]{Harris63}, and Athreya \& Ney \cite[Chap. IV]{AthreyaNey72}).
\begin{theo}
\label{th:expgrowthiid}
Assume that the lifetimes $T_v$ are
i.i.d. copies of an almost surely positive random variable $T$ with
non lattice distribution. Then:
\begin{equation}
\lim_{n\to\infty} \ee^{-\nu t} N_t = W\qquad \mbox{a.s.,}
\end{equation}
where:
\begin{itemize}
\item $W$ is a random variable with expectation $C$ and finite variance,
\item the growth rate (also called Malthusian parameter) $\nu$ is such that:
\begin{equation}
\label{iidmalthus}
2\EE[\ee^{-\nu T}]=1\;,
\end{equation}
\item the proportionality constant $C$ is:
\begin{equation}
\label{iidpropor}
C= \left(4\nu \EE[T\ee^{-\nu T}]\right)^{-1}\;.
\end{equation}
\end{itemize}
\end{theo}
The aim of this paper is to extend Theorem \ref{th:expgrowthiid}
to models in which lifetimes may be dependent, and in particular to
Bifurcating Markov Chains (BMC). Our main result (Theorem
\ref{th:expgrowthBMC}) generalizes Theorem \ref{th:expgrowthiid} to the
case where $(T_v)_{v\in\TT}$ is a multiplicatively ergodic, stationary
BMC. The growth rate $\nu$ and the proportionality constant
$C$ in that case are related to the
multiplicative ergodicity coefficients of
birth dates.
\vskip 2mm
Applications of branching process to cell lineage studies have a long history
(see e. g. \cite{KimmelAxelrod02} and references therein). 
Independence of lifetimes was
questioned very early: see \cite{Kendall52a}. Indeed, actual
data show two types of correlation \cite{Wangetal10}:
between the lifetimes of a
mother and its two daughters, and between the two sisters
conditioning on the mother; they will be referred to as
mother-correlation and sister-correlation.
It was remarked long ago by Powell \cite{Powell56}
that sister-correlations do not influence
exponential growth (see also \cite{CrumpMode69a,Harvey72} and
\cite[section 28.2 p.~ 158]{Harris63}). The effect of
mother-correlation on growth rates
was discussed by Harvey in \cite{Harvey72}. Since then, many models
have been proposed to account for ancestry dependence, in
particular by Smith \& Martin \cite{SmithMartin73}, L\"uck \& L\"uck 
\cite{Luck76},
Brook \cite{Brook81}, or Murphy et
al. \cite{Murphyetal84}: see \cite{Markhametal10,Nordonetal11}.
Here, lifetimes are seen as a stochastic process indexed by the binary
tree; see Pemantle \cite{Pemantle95}
and Benjamini \& Peres \cite{BenjaminiPeres94} as general references on
tree-indexed processes. Under a minimal hypothesis of stationarity,
exponential growth for the mean population size $\EE[N_t]$ is proved
and the growth rate $\nu$ as well as the proportionality constant $C$
are expressed in terms of the Laplace transforms of cell birth dates
(Theorem \ref{th:expgrowthSBP}).
Asymptotics of Laplace transforms for partial sums of a
Markov chain are usually described by multiplicative ergodicity
properties, which have been thoroughly studied by Meyn and his
co-workers
\cite{BalajiMeyn00,KontoyannisMeyn03,KontoyannisMeyn05,Meyn06}; see
also \cite[p.~519]{MeynTweedie09} for a short introduction.
It is therefore natural to use a BMC as a model of
lifetimes: see Benjamini \& Peres \cite{BenjaminiPeres94} for
tree-indexed Markov chains, Hwang \& Basawa \cite{HwangBasawa09} for
more asymptotic results, and Guyon \cite{Guyon07} for applications to cell
lineage data. Under a multiplicative ergodicity condition, Theorem
\ref{th:expgrowthiid} is generalized: $\ee^{-\nu t}N_t$ is shown to
converge almost surely; moreover, the growth
rate $\nu$ and the proportionality constant $C$ are explicitly
related to the multiplicative ergodicity coefficients
(Theorem \ref{th:expgrowthBMC}).
The proof follows a classical scheme,
already used by Bellman \& Harris for the i.i.d. case in
\cite{BellmanHarris52}. It consists of studying the first and second
moments of $N_t$, then prove convergence in quadratic mean,
and finally deduce almost sure convergence. This is related to what
Pemantle calls the ``second-moment method''
\cite[section 2.3]{Pemantle95}. In applying it, we have tried to give
the weakest possible conditions at each step, starting with the
stationarity hypothesis of
Theorem \ref{th:expgrowthSBP}. Proposition
\ref{prop:expgrowthL2} gives sufficient conditions that ensure quadratic
convergence of $\ee^{-\nu t}N_t$, Proposition \ref{prop:expgrowthas} gives
conditions for almost sure convergence. These conditions will be shown
to hold under the hypotheses of Theorem \ref{th:expgrowthBMC}.
\vskip 2mm\noindent
An obvious drawback for applications
is that the growth rate $\nu$ and the proportionality constant
$C$ cannot be computed in general. Therefore an
explicit model, potentially adjustable to observed data and for which
$\nu$ and $C$ can be computed in terms of the transition kernel, had
to be proposed. It was constructed as a quadratic
transformation of a bifurcating
autoregressive process
\cite{CowanStaudte86,Guyon07,DelmasMarsalle10,DeSaportaetal11}.
It depends on 5 identifiable parameters, (location, scale, and
shape for lifetime distribution plus mother-  and 
sister-correlations) and can be
fitted to actual data.
\vskip 2mm\noindent
Having in mind the application to cell lineage studies, it was natural
to write the results for the binary tree. Nevertheless, they extend
quite straightforwardly to processes on the infinite complete $k$-ary
tree for $k>2$, at the only expense of heavier notations. Remarks in
the text will make the generalization more precise. Further
extensions are possible, firstly to the case where $\TT$ is
a supercritical Galton-Watson tree and the lifetimes of daughters are
independent conditionally on their common mother, secondly to the case
where cell deaths are modelled by a binary process as in
\cite{DelmasMarsalle10}. They will be the
object of future work. 
\vskip 2mm\noindent
The paper is organized as follows. In section \ref{PSBP} the branching
process $(N_t)_{t\geqslant 0}$ associated to a bifurcating lifetime process
$(T_v)_{v\in\TT}$ is defined. Two notions of stationarity along
lineages are introduced
and the exponential growth of $\EE[N_t]$ is proved.
Section \ref{BMC} is devoted to the definition of a BMC, and
the statement of Theorem \ref{th:expgrowthBMC}. The explicit example of
a BMC for which the multiplicative ergodicity coefficients
can be computed, is presented in section \ref{explicit}.
The relation between mother-correlation and growth rate for a fixed
marginal distribution of lifetimes is discussed in section
\ref{associated}. Section \ref{L2as} is
devoted to conditions under which $\ee^{-\nu t}N_t$ converges
in $L^2$ and almost surely. These conditions are verified for a
multiplicatively ergodic
BMC in section \ref{proofBMC}.
\section{Stationary bifurcating processes}
\label{PSBP}
In this section, notations on bifurcating processes are introduced.
The birth date process $(S_v)_{v\in\TT}$ and the branching process
$(N_t)_{t\geqslant 0}$ associated to a bifurcating process
$(T_v)_{v\in\TT}$ are defined, and related by Lemma
\ref{lem:NtSv}.
Two notions of stationarity are introduced: birth-stationarity
(Definition \ref{def:birthstationary}) is the stationarity of birth
dates in a given generation; fork-stationarity
(Definition \ref{def:forkstationary}) is the stationarity of
couples of birth dates when the generations of the two cells and their
most recent common ancestor are fixed. Under birth-stationarity
the expectation of $N_t$ is proved to grow exponentially, and the parameters
of exponential growth $\nu$ and $C$ are related to the Laplace
transforms of birth dates (Theorem
\ref{th:expgrowthSBP}).
\vskip 2mm\noindent
Some classical notations for infinite trees will
be recalled first: see Pemantle \cite{Pemantle95}. The infinite
rooted complete binary tree is denoted by
$\TT$ and its root by $0$. If $v$ is a vertex of $\TT$, the number
of edges connecting $v$ to the root is denoted by $|v|$. If $v$ and
$w$ are two vertices of $\TT$, $v\preccurlyeq w$ is the order relation
that holds if
$v$ is in the path from $0$ to $w$; $v\wedge w$ is the most recent common
ancestor of $v$ and $w$,  i.e. the vertex at which the paths from $0$ to
$v$ and $w$ diverge.  If $v\neq 0$,
$\tilde{v}$ is the vertex such that $\tilde{v}\preccurlyeq v$ and
$|\tilde{v}|=|v|-1$ (referred to as the mother of $v$).
For $n\geqslant 0$, the $n$-th generation
$\Gamma_n$ is the set of vertices
$v$ such that $|v|=n$ (vertices at distance $n$ from the root). One simple way
to explicitly construct $\TT$ is to identify $\Gamma_n$ to the set of
binary vectors of length $n+1$, with first coordinate $0$.
With that identification, $v\preccurlyeq  w$ iff $v$
coincides with the $|v|+1$ first coordinates of $w$. The mother of
$v$, $\tilde{v}$ is
deduced from $v$ by removing its last coordinate. The two daughters of
$v$ are obtained by appending to $v$ a new coordinate $0$ or
$1$: they will be denoted by $v0$ and $v1$. The concatenation of $n$
zeros will be denoted by $0^{n}\in \Gamma_{n-1}$.
Besides
algorithmic considerations, one advantage of this construction is to
naturally endow $\TT$ with the alphabetical order.
\vskip 2mm
A bifurcating process is a set of almost surely positive random variables
$(T_v)_{v\in\TT}$ indexed by the binary tree $\TT$: $T_v$ is the
lifetime of cell $v$. The birth date process $(S_v)_{v\in \TT}$ is
also a bifurcating process: $S_v$ is the sum of
cell lifetimes from $0$ to $\tilde{v}$.
The branching
process $(N_t)_{t\geqslant 0}$ is the counting process of living cells
at time $t$.
\begin{defi}
\label{def:NtSv}
Let $(T_v)_{v\in \TT}$ be a bifurcating process.
\begin{enumerate}
\item For $v\in \TT$,
The birth date of cell $v$ is defined by $S_0=0$ and for $|v|>0$:
\begin{equation}
\label{defSv}
S_v=S_{\tilde{v}}+T_{\tilde{v}}\;.
\end{equation}
\item For $t\geqslant 0$, the number of living cells at time $t$ is
  defined by:
\begin{equation}
\label{defNt}
N_t = \sum_{v\in\TT} \II_{S_v\leqslant t} -\sum_{v\in
  \TT}\II_{S_{v0}\leqslant t}\;,
\end{equation}
where $\II_A$ denotes the indicator of event $A$.
\end{enumerate}
\end{defi}
If $S_v$ is the birth date of cell $v$, the common birth date of its
two daughters $S_{v0}=S_{v1}$ is also the death date of $v$. So
(\ref{defNt}) expresses the fact that
cells alive at time $t$ are the set difference of cells born
no later than $t$ with cells dead no later than $t$. A simpler
expression will be used.
\begin{lemm}
\label{lem:NtSv}
With the notations above,
\begin{equation}
\label{relNtSv}
N_t = \frac{1}{2}+\frac{1}{2} \sum_{v\in\TT} \II_{S_v\leqslant t}\;.
\end{equation}
\end{lemm}
\begin{proof}
From (\ref{defNt}), and using the relation $S_{v0}=S_{v1}$,
\begin{eqnarray*}
N_t&=& \sum_{v\in \TT} \II_{S_v\leqslant t}-
\sum_{v\in \TT} \II_{S_{v0}\leqslant t}\\
&=& \sum_{v\in\TT} \II_{S_v\leqslant t}
-\frac{1}{2} \sum_{\substack{w\in\TT\\w\neq 0}}\II_{S_w\leqslant t}\\
&=& 1+\frac{1}{2} \sum_{\substack{v\in\TT\\v\neq 0}}
\II_{S_v\leqslant t}\;,
\end{eqnarray*}
hence (\ref{relNtSv}).
\end{proof}
\begin{rem}
On the $k$-ary tree, (\ref{relNtSv}) becomes:
$$
N_t = \frac{1}{k}+\frac{k-1}{k} \sum_{v\in\TT} \II_{S_v\leqslant t}\;.
$$
\end{rem}
Consider the particular case where lifetimes in a given
generation are constant:
$$
\forall v\in \Gamma_n\;,\quad T_v = T_{0^{n+1}}\;,
$$
Denote by $S_n$ the common birth date of all cells in generation
$\Gamma_n$, and assume a law of large numbers is satisfied.
$$
\lim_{n\to \infty} \frac{S_n}{n} = \bar{t}>0\quad\mbox{a.s.}
$$
The rank of the generation alive at time $t$, denoted by
$G_t$, is the counting process associated to the sequence
$(S_n)_{n\in\NN}$, and $N_t=2^{G_t}$. Since $N_t$
doubles at $S_n$, $\ee^{-\nu t} N_t$ never converges, although
$$
\lim_{t\to \infty} \frac{\log(N_t)}{t} = \frac{\log(2)}{\bar{t}}\quad\mbox{a.s.}
$$
Consider now
$$
\frac{\log(\EE[N_t])}{t} = \frac{\log(\EE[\ee^{G_t \log 2}])}{t}\;.
$$
The convergence of $\frac{\log(\EE[\ee^{\theta G_t}])}{t}$ is a
G\"artner-Ellis condition on $G_t$. Glynn and Whitt
\cite{GlynnWhitt94} have proved that it is equivalent to the analogous
condition on $S_n$. But the convergence of $\frac{\log(\EE[N_t])}{t}$
does not imply that of $\ee^{-\nu t}\EE[N_t]$ (cf.
the case where all lifetimes are equal to some constant). A law of
large numbers, even strengthened by large deviations inequalities, does
not suffice to prove our results: additional hypotheses are needed. We
begin with stationarity requirements.
\vskip 2mm
The notion of stationarity that seems the most natural is
invariance through automorphisms of the tree,
as in Pemantle \cite{Pemantle92}. It will be satisfied by the BMC
models of the next two sections. Weaker hypotheses will
suffice for our preliminary convergence results.
The first one says that birth dates of cells in a given
generation have the same distribution. For $n\geqslant 0$, we shall
denote by $S_n$ the birth date of the first cell in generation
$\Gamma_n$, by alphabetical order.
$$
S_n = S_{0^n}=T_0+T_{0^2}+\cdots+T_{0^n}\;.
$$
\begin{defi}
\label{def:birthstationary}
The bifurcating process $(T_v)_{v\in \TT}$ is birth-stationary if
for all $n\in\NN$ and for all $v\in\Gamma_n$:
$$
S_v\mathop{=}^{\mathcal{D}}\; S_n\;.
$$
\end{defi}
Observe that birth-stationarity does not imply that lifetimes $T_v$ are
identically distributed, even in a given generation.
It will be used to prove the Ces\`aro convergence of
$\EE[\ee^{-\nu  t}N_t]$ in Theorem
\ref{th:expgrowthSBP} below. For the convergence in quadratic mean and
almost sure, a stronger notion will be used: the
joint distribution of the birth dates of
two cells in generations $n+i$ and $n+j$ with most recent common
ancestor in generation $n$, should depend only on $n$, $i$, and $j$. The
first such couple in alphabetical order is $(0^{n+i},0^n10^{j-1})$. The
corresponding birth dates will be denoted by $S_{n,i}^{(0)}$ and
$S_{n,j}^{(1)}$ (Figure \ref{fig:S01}):
\begin{equation}
\label{S0niS1nj}
S_{n,i}^{(0)}=S_{n+i} = S_{0^{n+i}}
\quad\mbox{and}\quad
S_{n,j}^{(1)}= S_{0^{n}10^{j-1}}\;.
\end{equation}
\begin{figure}[!ht]
\centerline{
\includegraphics[width=12cm]{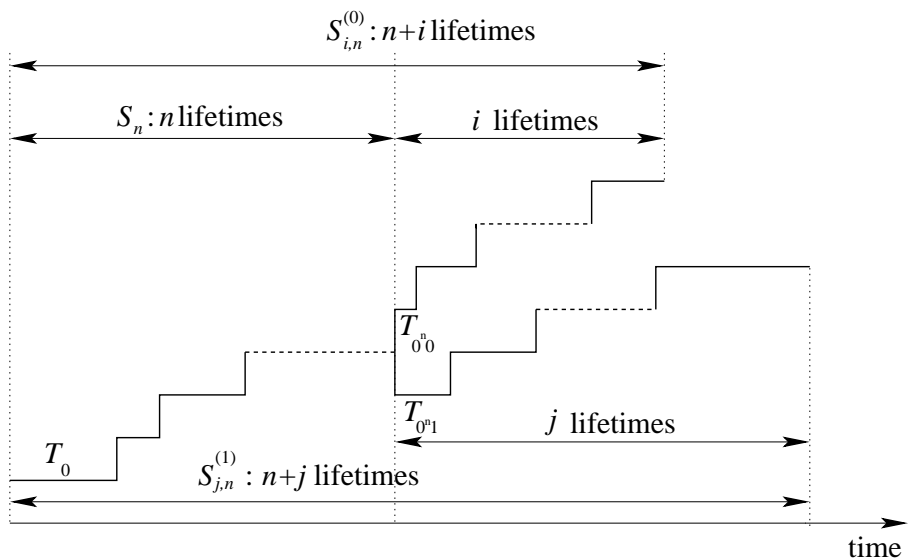}
} 
\caption{Birth dates $S_{n,i}^{(0)}$ and $S_{n,j}^{(1)}$ of the first
  two cells of generations $n+i$ and $n+j$, whose most recent common
  ancestor is in generation $n$.} 
\label{fig:S01}
\end{figure}

\begin{defi}
\label{def:forkstationary}
The bifurcating process $(T_v)_{v\in \TT}$ is fork-stationary if
for all $(n,i,j)\in\NN\times\NN\times \NN^*$ and for all
$(v,w)\in \Gamma_{n+i}\times\Gamma_{n+j}$
such that $v\wedge w\in \Gamma_n$:
$$
(S_v,S_w)\mathop{=}^{\mathcal{D}}\; (S_{n,i}^{(0)},S_{n,j}^{(1)})\;.
$$
\end{defi}
For $i=0$ the definition includes the case $v\preccurlyeq
w$: for all $n,k\in\NN\times \NN^*$, for all $(v,w)\in \Gamma_n\times
\Gamma_{n+j}$ such that $v\preccurlyeq w$,
$$
(S_v,S_w)\mathop{=}^{\mathcal{D}}\; (S_n,S_{n+j})\;.
$$
In particular all couples $(S_v,S_{v0})$ are
identically distributed  for $v\in\Gamma_n$, hence
lifetimes in a given generation have the same
distribution.
\begin{rem}
On the $k$-ary tree, one might think that forks should have $k$ teeth. Yet
fork-stationarity is used in Proposition \ref{prop:expgrowthL2} to express
$\EE[N_tN_{t+\tau}]$ in terms of the joint distribution of
$(S_{n,i}^{(0)},S_{n,j}^{(1)})$. This remains the same on the $k$-ary tree.
\end{rem}
\vskip 2mm\noindent
The main result of this section concerns the exponential growth
of $\EE[N_t]$; it relates the growth rate $\nu$ and
the proportionality constant $C$ to the Laplace transform of
$S_n$.
In order to enhance the link with Theorem \ref{th:expgrowthiid},
we chose to express our results in terms of Laplace transforms,
instead of characteristic functions or
logarithmic moment generating functions, as is customary
in large deviations theory (see e.g. \cite{DemboZeitouni98}).
Throughout the paper, the Laplace transforms evaluated at
$\gamma\geqslant 0$ of $S_n$, and of
$S_n$ conditioned on $T_0=u$,
will be denoted by
$\mathcal{L}_n(\gamma)$ and $\mathcal{L}_n(\gamma,u)$:
$$
\mathcal{L}_n(\gamma) = \EE[\ee^{-\gamma S_n}]
\quad\mbox{and}\quad
\mathcal{L}_n(\gamma,u) = \EE[\ee^{-\gamma S_n}\,|\,T_0=u]\;.
$$
\begin{theo}
\label{th:expgrowthSBP}
Let $(T_v)_{v\in\TT}$ be a birth-stationary bifurcating process. Assume
that $\nu$ and $C$ given below are well defined, positive, and finite.
\begin{equation}
\label{genmalthus}
\nu := \inf\{\gamma>0\,,\;\sum_{n=1}^{\infty}2^n\mathcal{L}_n(\gamma) <\infty\,\}\;,
\end{equation}
\begin{equation}
\label{genpropor}
C := \lim_{\gamma\searrow 0}
\frac{\gamma}{\gamma+\nu}\sum_{n=1}^{\infty}
2^{n-1}\mathcal{L}_n(\gamma+\nu)\;.
\end{equation}
Then for all $t\geqslant 0$,  $\EE[N_t]<\infty$ and:
\begin{equation}
\label{eq:cvexpectation}
\lim_{t\to\infty} \frac{1}{t}\int_0^{t}\ee^{-\nu s}\EE[N_s]\,\dd s = C\;.
\end{equation}
\end{theo}
In the particular case where the lifetimes $T_v$ are i.i.d. random
variables, $\mathcal{L}_n(\gamma)=(\EE[\ee^{-\gamma T_0}])^n$;
it can be easily checked that (\ref{genmalthus}) and (\ref{genpropor})
reduce to (\ref{iidmalthus}) and (\ref{iidpropor}).
\begin{proof}
From (\ref{relNtSv}):
\begin{eqnarray*}
N_t&=& \frac{1}{2}+\frac{1}{2} \sum_{v\in\TT} \II_{S_v\leqslant t}\\
&=& 1+\frac{1}{2}
\sum_{n=1}^{\infty}\sum_{v\in\Gamma_n} \II_{S_v\leqslant t}\;.
\end{eqnarray*}
By birth-stationarity:
\begin{equation}
\label{ENt}
\EE[N_t] = 1+\sum_{n=1}^{\infty} 2^{n-1}\PP[S_n\leqslant t]\;.
\end{equation}
By Markov's inequality, for all $\gamma>0$ and $t\geqslant 0$,
$$
\PP[S_n\leqslant t] \leqslant \ee^{\gamma t}\EE[\ee^{-\gamma S_n}]\;.
$$
The hypotheses of Theorem \ref{th:expgrowthSBP} imply that there
exists $\gamma>\nu$ such that
$$
\sum_{n=1}^{\infty} 2^{n-1}\EE[\ee^{-\gamma S_n}]<\infty\;.
$$
Hence $\EE[N_t]<\infty$ for all $t\geqslant 0$.
Consider now:
$$
A_\nu(t) = \ee^{-\nu t}\sum_{n=1}^{\infty}2^{n-1}\PP[S_n\leqslant t]\;.
$$
Let $\tilde{A}_\nu(\gamma)$ be the Laplace transform of
$A_\nu$.
The Laplace transform of $\PP[S_n\leqslant t]$,
evaluated at $\gamma>0$ is
$\frac{1}{\gamma}\mathcal{L}_n(\gamma)$. The Laplace transform of
$\ee^{-\nu t}\PP[S_n\leqslant t]$
is $\frac{1}{\gamma+\nu}\mathcal{L}_n(\gamma+\nu)$. Therefore:
$$
\tilde{A}_\nu(\gamma)=
\frac{1}{\gamma+\nu}\sum_{n=1}^{\infty}2^{n-1}\mathcal{L}_n(\gamma+\nu)\;.
$$
By (\ref{genpropor}):
$$
\lim_{\gamma\searrow 0} \gamma \tilde{A}_\nu(\gamma)=C\;.
$$
If both $\displaystyle{\lim_{t\to +\infty}A_\nu(t)}$ and
$\displaystyle{\lim_{\gamma\searrow 0}\gamma \tilde{A}_\nu(\gamma)}$ exist,
 the fact that they are equal is a
well known basic result of Laplace transform theory, known as
the Final Value Theorem. Deducing that the former limit exists from the
existence of the latter requires a Tauberian theorem: see
Feller \cite[section XIII.5]{Feller71} or Korevaar \cite{Korevaar04}. 
As a particular case of \cite[Theorem 2 p.~445]{Feller71}:
$$
\lim_{\gamma\searrow 0} \gamma \tilde{A}_\nu(\gamma)=C
\;\Longleftrightarrow\;
\lim_{t\to +\infty} \frac{1}{t}\int_0^{t} A_\nu(s)\,\dd s = C\;,
$$
which is the announced result.
\end{proof}
\begin{rem}
On the $k$-ary tree, (\ref{genmalthus}) and (\ref{genpropor}) become:
$$
\nu = \inf\{\gamma>0\,,\;\sum_{n=1}^{\infty}k^n\mathcal{L}_n(\gamma) <\infty\,\}\;,
$$
and
$$
C = \lim_{\gamma\searrow 0}
\frac{\gamma(k-1)}{\gamma+\nu}\sum_{n=1}^{\infty}
k^{n-1}\mathcal{L}_n(\gamma+\nu)\;.
$$
The proof is the same.
\end{rem}
Without any further assumption, nothing
more can be obtained than the Ces\`aro convergence
(\ref{eq:cvexpectation}), 
as the example of constant lifetimes shows.
To conclude that $\lim A_\nu(t)=C$ in the i.i.d. case, Bellman and Harris
\cite{BellmanHarris52} use Ikehara's Tauberian theorem. For the BMC
case, we shall need not only a limit, but also an
exponential speed of convergence. Although we have not found Lemma
\ref{lem:residue} in the
literature, it cannot be considered as new; it is closely
related to a large corpus of results going back to Haar, Wiener,
and Ikehara: see \cite{Korevaar04} as a
general reference, \cite{Feller41,Nakagawa05} for similar results.

\begin{lemm}
\label{lem:residue}
Let $f$ be a function and $\tilde{f}$ its Laplace transform. Suppose
that there exist two positive reals $\delta$ and $\epsilon$ such that:
\begin{enumerate}
\item $\tilde{f}$ is analytic in 
$\{z=x+\ii y\,,\;|x|<\delta+\epsilon\}\setminus \{0\}$,
\item 
$\tilde{f}$ has a simple pole at $0$, with residue $C$,
\item
$$
\int_{-\infty}^{\infty}|\tilde{f}(\delta+\ii y)|\,\dd y<\infty\;,
$$
\item
$$
\lim_{y\to \pm \infty}\tilde{f}(x+\ii y)=0\;,
$$
uniformly in $x\in [-\delta,\delta]$,
\item
$$
\psi:= \int_{-\infty}^{\infty}|\tilde{f}(-\delta+\ii y)|\,\dd y<\infty\;.
$$
\end{enumerate}
Then, for all $t>0$,
$$
|f(t)- C|\leqslant  \frac{\psi}{2\pi}\, \ee^{-\delta t} \;.
$$

\end{lemm}
\begin{proof}
By the inversion formula,
$$
f(t)= \frac{1}{2\pi\ii}
\lim_{\beta\to\infty}
\int_{\delta-\ii\beta}^{\delta+\ii\beta}\tilde{f}(\gamma)\,\ee^{\gamma t}\,\dd \gamma\;.
$$
Let $\mathcal{C}$ be the closed rectangular contour 
linking the points  $\delta-\ii\beta$, $\delta+\ii\beta$,
$-\delta+\ii\beta$, $-\delta-\ii\beta$. 
This contour encloses the simple pole at $\gamma=0$.
By the residue theorem,
$$
\int_{\mathcal{C}} \tilde{f}(\gamma)\,\ee^{\gamma t}\,\dd \gamma= 
2\pi\ii\,  \mathop{{\mbox{res}}}_{\gamma=0}(\tilde{f}(\gamma)\,\ee^{\gamma t})\;,
$$
 where
$$
\mathop{{\mbox{res}}}_{\gamma=0}(\tilde{f}(\gamma)\,\ee^{\gamma t})=
\lim_{\gamma\to 0} \gamma\tilde{f}(\gamma)\,\ee^{\gamma t}=C\;.
$$
Consequently,
$$
\frac{1}{2\pi\ii}\int_{\delta-i\beta}^{\delta+i\beta}\tilde{f}(\gamma)
\,\ee^{\gamma t}\dd \gamma-C= I_1+I_2+I_3\;,
$$
with:
\begin{eqnarray*}
I_1&=&\displaystyle{ -\frac{1}{2\pi\ii}
\int_{\delta+\ii\beta}^{-\delta+\ii\beta}\tilde{f}(\gamma)
\,\ee^{\gamma  t}\,\dd \gamma\;,}\\[2ex]
I_2&=&\displaystyle{ -\frac{1}{2\pi\ii}
\int_{-\delta+\ii\beta}^{-\delta-\ii\beta}
\tilde{f}(\gamma)\,\ee^{\gamma t}\,\dd \gamma\;,}\\[2ex]
I_3&=&\displaystyle{ -\frac{1}{2\pi\ii}
\int_{-\delta-\ii\beta}^{\delta-\ii\beta}\tilde{f}(\gamma)\,\ee^{\gamma t}\dd \gamma\;.}
\end{eqnarray*}
By condition \textit{4}, $I_1$ and $I_3$ tend to $0$ as $\beta$ tends
to infinity. Therefore:
\begin{eqnarray*}
|f(t)-C|&\leqslant&\displaystyle{ 
\frac{1}{2\pi}\left|\int_{-\infty}^{+\infty}
\tilde{f}(-\delta +\ii  y)\,\ee^{(-\delta+\ii y)t}\,\dd y\right|}\\[2ex]
&\leqslant&\displaystyle{
\int_{-\infty}^{+\infty}
|\tilde{f}(-\delta +\ii  y)|\,\ee^{-\delta t}\,\dd y}\\[2ex]
&=&  \frac{\psi}{2\pi}\,\ee^{-\delta t}\;,
\end{eqnarray*}
by condition \textit{5}\;.
\end{proof}

\section{Bifurcating Markov chains}
\label{BMC}
Bifurcating Markov chains (BMC) were considered long ago, beginning with
Spitzer \cite{Spitzer75} in the binary valued case (see
Benjamini \& Peres \cite{BenjaminiPeres94} for further
reference). They were studied as
cell lineage models by Guyon \cite{Guyon07}.
As in the one-dimensional case, the probability distribution of a BMC
is determined by an initial measure and a transition kernel. Here is the
definition, adapted to our case (as usual, $\mathcal{B}$ denotes the
Borel $\sigma$-algebra).
\begin{defi}
\label{def:kernel}
A \emph{transition kernel} $P$ is a mapping defined on
$\RR^+\times\mathcal{B}(\RR^+\times\RR^+)$ such that:
\begin{itemize}
\item For all $B\in\mathcal{B}(\RR^+\times\RR^+)$, $t\mapsto P(t,B)$
  is $\mathcal{B}(\RR^+)$-measurable,
\item For all $t\in\RR^+$, $B\mapsto P(t,B)$ is a probability measure
  on $\mathcal{B}(\RR^+\times\RR^+)$.
\end{itemize}
\end{defi}
\begin{defi}
\label{def:BMC}
Let $\mu$ be a probability measure on $\RR^+$, and $P$ be a transition
kernel. The distribution of a BMC $(T_v)_{v\in\TT}$ with initial
measure $\mu$ and transition kernel $P$ is inductively
defined as follows.
\begin{itemize}
\item The distribution of $T_0$ is $\mu$.
\item For $n\geqslant 1$, $(T_w)_{w\in \Gamma_{n+1}}$ and
$(T_u)_{u\in\Gamma_0\cup\ldots\cup\Gamma_{n-1}}$ are independent conditionally
  upon $(T_v)_{v\in \Gamma_{n}}$.
\item For all $n\geqslant 0$, the conditional distribution of
$(T_w)_{w\in \Gamma_{n+1}}$ knowing\\ $(T_v)_{v\in \Gamma_{n}}$ is
defined for $(B_v)_{v\in \Gamma_{n}}\in\mathcal{B}(\RR^+\times\RR^+)$ by:
$$
\PP[\,\forall v\in\Gamma_n\,,\; (T_{v0},T_{v1})\in B_v\;|\;
\forall v\in\Gamma_n\,,\;T_v=t_v\,] =
\prod_{v\in\Gamma_n} P(t_v,B_v)\;.
$$
\end{itemize}
\end{defi}
In other words, given the lifetimes of mothers in generation $n$,
the lifetimes of couples of daughters in generation $n+1$ are drawn
independently, each according to the transition kernel $P$.
Unlike in
\cite{Guyon07}, $P$ must be symmetric to ensure
stationarity: for all $t\in\RR^+$, for all
$B\in \mathcal{B}(\RR^+)$,
\begin{equation}
\label{eq:symmetric}
P_0(t,B) := P(t,B\times\RR^+) = P(t,\RR^+\times B) =: P_1(t,B)\;.
\end{equation}
The initial measure $\mu$ is supposed to be invariant for
both marginal kernels.
\begin{equation}
\label{eq:invariant}
\forall B\in\mathcal{B}(\RR^+)\;,\quad
\int_{\RR^+} P_0(u,B)\,\dd \mu(u) = \mu(B)\;.
\end{equation}
Symmetry and invariance imply that the distribution of
$(T_v)_{v\in\TT}$ is automorphism invariant in the sense of
\cite{Pemantle92}. In particular it is birth- and fork-stationary, in the sense
of Definitions \ref{def:birthstationary} and \ref{def:forkstationary}.
\begin{rem}
On the $k$-ary tree, a transition kernel is a mapping defined on
$\RR^+\times \mathcal{B}((\RR^+)^k)$. The generalization of Definition
\ref{def:BMC} is straightforward: knowing the lifetimes of mothers in
generation $n$, the lifetimes for all $k$-tuples of daughters are
drawn independently according to the transition kernel.
\end{rem}
Let $(N_t)_{t\geqslant 0}$ be the branching process associated to $(T_v)_{v\in \TT}$
(Definition \ref{def:NtSv}). Our goal is to prove the
extension of Theorem
\ref{th:expgrowthiid}, i.e. the almost sure convergence of $\ee^{-\nu
  t} N_t$. Asymptotics on the Laplace transform of $S_n$ will be needed;
the expressions of $\nu$
and $C$  given in Theorem \ref{th:expgrowthSBP} suggest using multiplicative
ergodicity for the sums of lifetimes $S_n$: see
\cite{BalajiMeyn00,KontoyannisMeyn03,KontoyannisMeyn05,Meyn06}.
In order to enhance the similarity with the i.i.d. case, we chose
to express multiplicative ergodicity in a slightly different manner.
\begin{defi}
\label{def:multergo}
The sums $S_n$ are said to be multiplicatively ergodic if for all
$u\in\RR^+$, $\gamma\in\RR^+$, $n\in\NN^*$,
\begin{equation}
\label{eq:multergodicitycond}
\mathcal{L}_n(\gamma,u)
=\alpha(\gamma,u)L^n(\gamma)+r_n(\gamma,u)\;,
\end{equation}
where $\alpha$, $L$ and $r_n$ are such that:
\begin{enumerate}
\item the equation $2L(\gamma)=1$ has a unique solution
  denoted by $\nu$,
\item the mapping $L$ is derivable at $\nu$ and $L'(\nu)< 0$,
\item the series $\sum_n 2^nr_n(\gamma,u)$ converges
  uniformly in $\gamma$ in a neighborhood of $\nu$, uniformly in $u$,
\item the mappings $u\mapsto \alpha(\gamma,u)$ and $u\mapsto
  r_n(\gamma,u)$ are $\mu$-integrable, uniformly in the other
  variables,
\item the mapping $(y,z)\mapsto a(\nu,y)a(\nu,z)$ is 
$P(x,(y,z))$-integrable, uniformly in $x$.
\end{enumerate}  
\end{defi}
Observe that under Definition \ref{def:multergo}, for all $u$ the sum
$\sum_n2^n\mathcal{L}_n(\gamma,u)$ converges for $\gamma> \nu$,
diverges for $\gamma\leqslant \nu$. Therefore, the same holds for
$\sum_n2^n\mathcal{L}_n(\gamma)$, and 
the definition of $\nu$ by $2L(\nu)=1$ is coherent with 
(\ref{genmalthus}).
\vskip 2mm 
Theorem 4.1 p.~325 of Kontoyannis and Meyn \cite{KontoyannisMeyn03}
relates multiplicative ergodicity to geometric
ergodicity. More precise analyticity conditions
will be needed for the following function:
\begin{equation}
\label{Btu}
B_\nu(t,u) = \ee^{-\nu t}\sum_{n=1}^\infty 2^{n-1}\PP[S_n \leqslant t\,|\,T_0=u]\;. 
\end{equation}
Under (\ref{eq:multergodicitycond}), its Laplace transform is:
\begin{equation}
\tilde{B}_\nu(\gamma,u) =
\frac{\alpha(\nu+\gamma,u)}{\nu+\gamma}\frac{L(\nu+\gamma)}{1-2L(\nu+\gamma)}
+\frac{1}{\nu+\gamma}\sum_{n=1}^\infty 2^{n-1}r_n(\nu+\gamma,u)
\;.
\end{equation}
Our hypotheses will be the following.
\begin{itemize}
\item[$(\mathcal{C}_1)$] For all $u>0$, $B_\nu(t,u)$ and $\tilde{B}_\nu(\gamma,u)$
  satisfy the hypotheses of Lemma \ref{lem:residue}, for some
  $\delta>0$ (not depending on $u$), and 
$$
C(u) = \lim_{\gamma\to 0} \gamma \tilde{B}_\nu(\gamma,u)
=-\frac{\alpha(\nu,u)}{4\nu L'(\nu)}\;.
$$
Let 
$$
\psi(u) = \int_{-\infty}^{\infty}|\tilde{B}_\nu(-\delta+\ii y,u)|\,\dd y\;.
$$
\item[$(\mathcal{C}_2)$] The mapping $(y,z)\mapsto
\alpha(\nu,y)\psi(z)$ is $P(x,(y,z))$-integrable, uniformly in $x$.  
\end{itemize}
Admittedly, $(\mathcal{C}_1)$ and $(\mathcal{C}_2)$ are not easy to
verify, unless an explicit expression of $\mathcal{L}_n(\gamma,u)$ is
available. This will be the case for the model to 
be presented in the next section.
\vskip 2mm\noindent
Using point \textit{4} of Definition \ref{def:multergo}, let
$$
\alpha(\gamma) = \int_{\RR^+} \alpha(\gamma,u)\,\dd \mu(u)
\quad\mbox{and}\quad
r_n(\gamma) = \int_{\RR^+} r_n(\gamma,u)\,\dd \mu(u)\;.
$$
Then:
\begin{equation}
\label{eq:multergodicity}
\mathcal{L}_n(\gamma)
=\alpha(\gamma)L^n(\gamma)+r_n(\gamma)\;.
\end{equation}
The proportionality constant $C$
naturally relates to $\alpha(nu)$ and $L(\nu)$.
\begin{theo}
\label{th:expgrowthBMC}
Assume that symmetry (\ref{eq:symmetric}), invariance
(\ref{eq:invariant}), 
multiplicative ergodicity (\ref{eq:multergodicity}) hold together with
$(\mathcal{C}_1)$ and $(\mathcal{C}_2)$.
Then
\begin{equation}
\lim_{n\to\infty} \ee^{-\nu t} N_t = W\qquad \mbox{a.s.,}
\end{equation}
where:
\begin{itemize}
\item $W$ is a random variable with expectation $C$ and finite variance,
\item the growth rate $\nu$ is such that:
\begin{equation}
\label{bmcmalthus}
2L(\nu)=1\;,
\end{equation}
\item the proportionality constant $C$ is:
\begin{equation}
\label{bmcpropor}
C=  -\frac{\alpha(\nu)}{4\nu L'(\nu)}\;.
\end{equation}
\end{itemize}
\end{theo}
\begin{rem}
On the $k$-ary tree, (\ref{bmcmalthus}) and (\ref{bmcpropor}) become:
$$
k L(\nu)=1\;,
\quad\mbox{and}\quad
C=  -\frac{(k-1)\alpha(\nu)}{k^2\nu L'(\nu)}\;.
$$
\end{rem}
As mentioned in the introduction, we have tried to give the weakest
possible conditions to ensure convergence in $L^2$ on the one hand
(Proposition \ref{prop:expgrowthL2}), almost sure convergence on the other
hand (Proposition \ref{prop:expgrowthas}). The proof of
Theorem \ref{th:expgrowthBMC} will be completed in section
\ref{proofBMC}.
\section{An explicit model}
\label{explicit}
In this section an explicit example for the
result of the previous section is constructed: a BMC with prescribed
invariant measure $\mu$, for which
symmetry (\ref{eq:symmetric}), invariance (\ref{eq:invariant}), and
multiplicative ergodicity (\ref{eq:multergodicity}) hold. The
model depends on an identifiable set of parameters,
potentially adjustable to observed data.
\vskip 2mm\noindent
The construction of stationary
processes with prescribed marginal distribution has been the object of
many studies: see Pitt et al. \cite{Pittetal02} and references
therein. We shall follow a simple approach, first constructing
a bifurcating autoregressive process, then transforming it to obtain the
desired marginals. Bifurcating autoregressive (BAR) processes
were introduced by Cowan and Staudte \cite{CowanStaudte86} precisely
as cell lineage models. They have been extensively studied
since, and the problem of parameter estimation has recently received a
lot of attention
\cite{Guyon07,HwangBasawa11,DelmasMarsalle10,DeSaportaetal11}.
Our model is similar to that of Guyon
\cite{Guyon07}. The construction
begins with a family
of i.i.d. random variables $(\epsilon_v)_{v\in \TT}$, each with standard
Gaussian $\mathcal{N}(0,1)$ distribution. Let $\rho_m$ and
$\rho_s$ be two reals in $(-1\,; 1)$; they will be the mother- and
sister-correlations of our BAR process. It is defined
inductively by $X_0=\epsilon_0$, and for all $v\in \TT$:
\begin{equation}
\label{AR1}
\left\{\begin{array}{lcl}
X_{v0}&=& \rho_m\, X_v+\sqrt{1-\rho_m^2} \,\epsilon_{v0}\\[2ex]
X_{v1}&=& \rho_m\, X_v+\sqrt{1-\rho_m^2}\,
(\rho_s\,\epsilon_{v0}+\sqrt{1-\rho_s^2}\,\epsilon_{v1})\;.
\end{array}\right.
\end{equation}
By construction, $(X_v)_{v\in \TT}$ is both a BMC on $\TT$ and a
Gaussian process. It is symmetric in the sense of (\ref{eq:symmetric})
and the standard Gaussian distribution is the invariant distribution of the
marginal kernel, in the sense of (\ref{eq:invariant}).
Let $f$ denote the composition of the quantile function of the
desired distribution $\mu$ on $\RR^+$ by the distribution function of the
$\mathcal{N}(0,1)$.
If $X$ follows the $\mathcal{N}(0,1)$,
then $T=f(X)$ has distribution $\mu$.
For all $v\in \TT$, let $T_v=f(X_v)$:
$(T_v)_{v\in \TT}$ is a BMC on $\TT$, for which
(\ref{eq:symmetric}) and (\ref{eq:invariant}) hold. Observe moreover
that $(X_v)$ converges at geometric speed along the rays of $\TT$,
hence multiplicative ergodicity holds for the birth dates of $(T_v)$,
by Theorem 4.1 p.~325 of 
\cite{KontoyannisMeyn03}. 
Of course the mother- and sister-correlations are not
$\rho_m$ and $\rho_s$ anymore. But they can
be computed in terms of $f$, $\rho_m$, and $\rho_s$, and so the
model can be adjusted to fit not only the observed distribution of
lifetimes but also estimated correlations.
\vskip 2mm\noindent
As remarked as early as 1932 by Rahn \cite{Rahn32},  actual
lifetime data show a unimodal right-skewed shape (see also Murphy et
al. \cite{Murphyetal84}). They
have been fitted by many types of
distributions: from Gamma and logbeta (Kendall
\cite{Kendall48}), to lognormal and reciprocal
normal (Kubitschek \cite{Kubitschek71}): see
John \cite{John81} and references therein.
The difficulty is to exhibit a realistic example where the hypotheses
of Theorem \ref{th:expgrowthBMC} hold, with explicitly computable
$\alpha$ and $L$. We propose to transform the standard Gaussian
variables $X_v$ of the BAR process defined by (\ref{AR1}), by the following
function, depending on three parameters:
$$
f(x) = a+b(x+c)^2\;.
$$
If $X$ is normally distributed, then
$(X+c)^2$ has a noncentral chi-squared distribution and the shape can
be adjusted by $c$; using the location and scale parameters $a$ and $b$, it
can be fitted to actual lifetime data.
The Laplace transforms of quadratic forms of
auto-regressive processes can be explicitly computed, using a technique
due to Klepsyna et al. \cite{Kleptsynaetal02}. The expression of the
Laplace transform $\mathcal{L}_n(\gamma)$
below is due to Alain Le Breton \cite{LeBreton13}.
\begin{prop}
\label{prop:lapalb}
Let $(\epsilon_n)_{n\in\NN}$ be a sequence of i.i.d. random variables,
with common distribution $\mathcal{N}(0,1)$. Let
$\rho=\rho_m\in(-1,1)$. Let $(X_n)_{n\in\NN}$ be
the stationary autoregressive chain defined by $X_0=\epsilon_0$ and
for $n\geqslant 0$:
$$
X_{n+1} = \rho\, X_n+\sqrt{1-\rho^2} \,\epsilon_{n+1}\;.
$$
Let
$$
S_n = \sum_{k=0}^n f(X_k) = \sum_{k=0}^n(a+b(X_k+c)^2)\;.
$$
Denote:
$$
\gamma_1 = 2\gamma b(1-\rho^2)\;,\quad
\gamma_2 = \frac{1-\rho}{\gamma_1+(1-\rho)^2}\;,
$$
$$
A = (1-\rho)\gamma_2\;,\quad
B = -2\rho\, \gamma_2^2\;,\quad
C = \frac{2 \rho \gamma_1}{1-\rho^2}\,\gamma_2^2\;,
$$
$$
\lambda_{\pm} = \frac{\gamma_1+1+\rho^2\pm
\sqrt{(\gamma_1+(\rho+1)^2)(\gamma_1+(\rho-1)^2)}}{2}\;.
$$
$$
\beta_+ = \frac{1-\lambda_-+\frac{\gamma_1}{1-\rho^2}}
{\lambda_+-\lambda_-}
\;;\quad
\beta_- = \frac{\lambda_+-1+\frac{\gamma_1}{1-\rho^2}}
{\lambda_+-\lambda_-}\;,
$$
$$
\pi_n = \beta_+\lambda_+^{n+1}+\beta_-\lambda_-^{n+1}\;,\quad
\psi_n = \beta_+\left(\frac{\lambda_+}{\rho}\right)^{n}
+\beta_-\left(\frac{\lambda_-}{\rho}\right)^{n}\;.
$$
Then:
\begin{equation}
\label{eq:tlalb}
\mathcal{L}_n(\gamma) = \ee^{-(n+1) a \gamma }\,(\pi_n)^{-1/2}\,
\exp\left(-\frac{c^2 \gamma_1}{2(1-\rho^2)}\Sigma_n\right)
\end{equation}
with
$$
\Sigma_n =
n\,A  + \frac{1}{\pi_0}
+B\left(\frac{\rho}{\pi_0}-\frac{\psi_n}{\psi_{n+1}}\right)
+C\left(\frac{\rho}{\pi_0}-\frac{1}{\psi_{n+1}}\right)\;.
$$
\end{prop}
The asymptotics of (\ref{eq:tlalb}) is easy to write,
because $\lambda_+$ and $\lambda_-$ are such that:
$$
\frac{\lambda_+}{\rho}>1\quad\mbox{and}\quad
\frac{\lambda_-}{\rho}<1\;.
$$
From there it follows that, as $n$ tends to $+\infty$,
$$
\mathcal{L}_n =\alpha(\gamma) L^n(\gamma) \Big(1+ O((\rho/\lambda_+)^n)\Big)\;,
$$
with
\begin{equation}
\label{alphas}
\alpha(\gamma) = (\beta_+\lambda_+)^{-1/2}
\exp\left(-\frac{c^2 \gamma_1}{2(1-\rho^2)}\left(
\frac{1}{\pi_0}
+B\left(\frac{\rho}{\pi_0}-\frac{\rho}{\lambda_+}\right)
+C\left(\frac{\rho}{\pi_0}\right)\right)\right)\;,
\end{equation}
and
\begin{equation}
\label{Ls}
L(\gamma)=\ee^{-a\gamma}(\lambda_+)^{-1/2}\exp\left(-\frac{\gamma bc^2(1-\rho)}
    {2\gamma b(1+\rho)+(1-\rho)}
\right)\;.
\end{equation}
Alain Le Breton \cite{LeBreton13} as also obtained an
analogous, though more complicated result for the conditional Laplace
transform $\mathcal{L}_n(\gamma,u)$ that will not be reproduced
here. From that result, it can be deduced that
the hypotheses of Theorem \ref{th:expgrowthBMC} are
satisfied. Actually, it can be checked that $\alpha(\gamma,u)$ and
$r_n(\gamma,u)$ are uniformly bounded in $u$ and $\gamma$ over
$(\RR^+)^2$, which considerably
simplifies conditions \textit{4} and \textit{5} of Definition
\ref{def:multergo} as well as conditions $(\mathcal{C}_1)$ and
$(\mathcal{C}_2)$.
\vskip 2mm\noindent
The growth rate $\nu$ and the proportionality constant $C$
can be derived
from (\ref{alphas}) and (\ref{Ls}), through (\ref{bmcmalthus}) and
(\ref{bmcpropor}). The growth rate $\nu$ is the solution of $2L(\nu)=1$. It has
no general explicit expression, but it can be numerically
computed. In the particular case $a=c=0$ ($f(x)=bx^2$), $\nu=\nu_0$ is
found to be:
$$
\nu_0 = \frac{3}{2b}\,\frac{1-\rho^2/4}{1-\rho^2}\;.
$$
In the general case, it can be checked that
$\nu\leqslant \nu_0$.
\section{Correlations and growth rate}
\label{associated}
As mentioned in the introduction, the influence of
lifetime correlations on the exponential growth of the
colony was discussed long ago
\cite{Powell56,CrumpMode69a,Harvey72}. That sister-correlations do not
change the exponential growth rate was remarked by all early
authors, and is confirmed by Theorem \ref{th:expgrowthSBP}. The
influence of mother-correlation is discussed here.
\vskip 2mm\noindent
The hypotheses in this section are those of Theorem
\ref{th:expgrowthSBP}: birth-stationarity and definition of $\nu$ and $C$
by (\ref{genmalthus}) and (\ref{genpropor}).
A general comparison result will first be obtained under
association hypotheses. Recall that a sequence of random
variables $(X_n)_{n\geqslant 0}$ is \emph{associated}
if for all $n$, the vector
$X^{(n)} = (X_1, X_2, \ldots, X_n)$ satisfies the following condition:
for any coordinatewise bounded and nondecreasing functions $f$, $g$ on
${\mathbb R}^n$, ${\mbox{Cov}} ( f ( X^{(n)}) , g ( X^{(n)}) )  \geqslant
0$.
We refer to Esary et al. \cite{Esaryetal67} for more about this
notion. The sequence  $(X_n)_{n\geqslant 1}$ is {\it negatively
  associated} if for any coordinatewise bounded and nondecreasing
functions $f$, $g$ defined respectively on
${\mathbb R}^{|I|}$, ${\mathbb R}^{|J|}$ where
$I$ and $J$ are disjoint subsets of $\NN$,
${\mbox{Cov}} ( f ( (X_i)_{i\in I}) , g ( (X_i)_{j\in J}))  \leqslant 0$.
This definition was introduced by Joag-Dev and Proschan
\cite{JoagDevProschan83}.
\begin{prop}\label{prop:associated}
Suppose that the sequence $(T_{0^n})_{n\geqslant 1}$ is associated
(respectively: negatively associated)
and that the hypotheses of Theorem \ref{th:expgrowthSBP} are
satisfied.
Let $(T_{0^n}^*)_{n\geqslant 1}$ be a sequence of independent random
variables such that $T_{0^n}^*$ and $T_{0^n}$ have the same
distribution. Let $\nu, \nu^*$ and $C, C^*$ be the respective growth
rates and proportionality constants corresponding to
$(T_{0^n})_{n\geqslant 1}$ and $(T_{0^n}^*)_{n\geqslant 1}$ through
(\ref{genmalthus}) and (\ref{genpropor}).
Then  $\nu^*\leqslant \nu$ and $C\leqslant C^*$
(respectively: $\nu\leqslant \nu^*$ and $C^*\leqslant C$).
\end{prop}
\begin{proof}
We only give the proof for the case of association, the case of negative
association is symmetric. Let
$S^*_n=T^*_0+T^*_{0^2}+\cdots+T^*_{0^{n}}$. If
$(T_{0^n})_{n\geqslant 1}$ is an associated sequence, then for any
positive real $\gamma$,
\begin{equation}
\label{stocomplaplace}
\EE[\ee^{-\gamma S^*_n}]\leqslant \EE[\ee^{-\gamma S_n}]\;.
\end{equation}
From there, it follows immediately that
$\nu^*\leqslant \nu$, by (\ref{genmalthus}). The inequality $C\leqslant
C^*$ then follows from the fact that $S_n$
stochastically dominates $S^*_n$: for all $t\geqslant 0$,
\begin{equation}
\label{stocompcdf}
\PP[S_n\leqslant t]\leqslant \PP[S^*_n\leqslant t]\;.
\end{equation}
Stochastic comparison results such as (\ref{stocomplaplace}) and
(\ref{stocompcdf}) are well known decoupling inequalities, and we
shall omit their proofs: see de la Pe\~na and Lai \cite[p.~
118]{PenaLai01} and Shao \cite{Shao00}.
\end{proof}
Proposition \ref{prop:associated} indicates that for a fixed marginal
distribution of lifetimes, the growth rate $\nu$ should increase as
the mother-correlation increases from $0$ to $1$. This is indeed what
can be observed on the explicit model of section \ref{explicit}.
In that model, $\nu$ is defined by $2L(\nu)=1$, where $L$ is given by
(\ref{Ls}). In (\ref{Ls})
$\rho$ is the correlation between successive
steps of the BAR process, which differs from the correlation between
the lifetimes of a mother and its daughter. The latter will be denoted by
$\varrho$. The expression of $\varrho$ as a function of the parameters
$a,b,c,\rho$ is easily calculated. It only depends on $c$ and $\rho$.
$$
\varrho=\mbox{Cor}(T_v,T_{v0})= \frac{\rho^2+2c^2\rho}{1+2c^2}\;.
$$
As $\rho$ increases from $0$ to $1$, so does $\varrho$. As $\varrho$
tends to  $+1$, $\nu$ tends to:
\begin{itemize}
\item
$+\infty$ if $a=0$,
\item
$\log(2)/a$ else.
\end{itemize}
The limit value $\log(2)/a$
is the growth rate that would be achieved if all lifetimes were equal to $a$, which
is the minimal value that a lifetime can take in the model.
\section{Convergence in quadratic mean and almost sure}
\label{L2as}
Conditions for the convergence of $\ee^{-\nu t}N_t$ are
given in this section. Under the hypothesis of fork-stationarity of
Definition \ref{def:forkstationary},
Proposition \ref{prop:expgrowthL2} below gives
a general condition under which $\ee^{-\nu t}N_t$
converges in $L^2$.
\begin{prop}
\label{prop:expgrowthL2}
Let $(T_v)_{v\in\TT}$ be a fork-stationary bifurcating process.
Assuming that the hypotheses of Theorem \ref{th:expgrowthSBP} hold,
let $\nu$ and $C$ be defined by (\ref{genmalthus}) and
(\ref{genpropor}).
For all $t,\tau\geqslant 0$, let:
\begin{equation}
\label{Sigma1t}
\Sigma_1(t)=
\sum_{n=0}^{\infty}(n+1)2^n\,\PP[S_n\leqslant t]\;,
\end{equation}
\begin{equation}
\label{Sigma2ttau}
\Sigma_2(t,\tau) = \sum_{n=0}^{\infty}\sum_{i=1}^{\infty}\sum_{j=1}^{\infty}
2^{n+i+j}\,
\PP[\, S_{n,i}^{(0)}\leqslant t \,,\;S_{n,j}^{(1)}\leqslant
t+\tau\,]\;.
\end{equation}
Assume that for all $t,\tau\geqslant 0$, $\Sigma_1(t)$ and
$\Sigma_2(t,\tau)$ are finite,
that the following limits exist and the
second one does not depend on $\tau$.
\begin{equation}
\label{limSigma1}
\lim_{t\to +\infty} \ee^{-2\nu t} \Sigma_1(t)=0\;,
\end{equation}
\begin{equation}
\label{limSigma2}
\lim_{t\to +\infty} \ee^{-\nu(2t+\tau)} \Sigma_2(t,\tau)
=C_2<+\infty\;.
\end{equation}
Then as $t$ tends to infinity,
$\ee^{-\nu t}N_t$ converges in quadratic mean to a random variable $W$
with expectation $C$.
\end{prop}
Observe that since $L^2$-convergence implies $L^1$-convergence,
$$
\lim_{t\to\infty}\EE[\ee^{-\nu t} N_t]
= \EE[W]=C>0\;.
$$
\begin{proof}
We first express the product $N_tN_{t+\tau}$ as a function of
birth dates. For this, recall the expression  (\ref{relNtSv}) of
$N_t$, given in Lemma \ref{lem:NtSv}:
$$
N_t = \frac{1}{2}+\frac{1}{2}
\sum_{v\in\TT} \II_{S_v\leqslant t}\;.
$$
Hence:
$$
(2N_t-1)(2N_{t+\tau}-1)= \sum_{(v,w) \in \TT^2}
\II_{S_v\leqslant  t}\,\II_{S_w\leqslant t+\tau}\;.
$$
For any couple $(v,w)\in\TT^2$, one and only one of
the following three cases occurs.
\begin{enumerate}
\item $w\preccurlyeq v$, in which case
$\II_{S_v\leqslant  t}\,\II_{S_w\leqslant t+\tau}=
\II_{S_v\leqslant  t}$,
\item $|v\wedge w|<\min\{|v|,|w|\}$,
\item $v\preccurlyeq w$ and $v\neq w$.
\end{enumerate}
Decomposing the sum over the three cases,
taking expectations on both sides and using
fork-stationarity:
\begin{equation}
\label{Sigma123}
\EE[(2N_t-1)(2N_{t+\tau}-1)] = \Sigma_1(t)+\Sigma_2(t,\tau)+\Sigma_3(t,\tau)\;,
\end{equation}
with
\begin{eqnarray*}
\Sigma_1(t)&=&
\sum_{n=0}^{\infty}(n+1)2^n\PP[S_n\leqslant t]\;,\\
\Sigma_2(t,\tau)&=&\sum_{n=0}^{\infty}\sum_{i=1}^{\infty}\sum_{j=1}^{\infty}
2^{n+i+j}\PP[S_{n,i}^{(0)}\leqslant t,\,S_{n,j}^{(1)}\leqslant t+\tau]\;,\\
\Sigma_3(t,\tau)&=&\sum_{n=0}^{\infty}\sum_{j=1}^{\infty}
2^{n+j}\PP[S_n\leqslant t,\,S_{n+j}\leqslant t+\tau]\;.
\end{eqnarray*}
From the hypotheses for all $t,\tau\geqslant 0$,
$\Sigma_1(t)$ and $\Sigma_2(t,\tau)$ are finite.
Remark that:
\begin{eqnarray*}
\Sigma_3(t,\tau)&\leqslant&
\sum_{n=0}^{\infty}\sum_{j=1}^{\infty}
2^{n+j}\PP[S_{n+j}\leqslant t+\tau]\\
&=&
\sum_{m=1}^{\infty}m2^{m}\PP[S_m\leqslant t+\tau]\\
&\leqslant&
\Sigma_1(t+\tau)\;.
\end{eqnarray*}
Therefore $\Sigma_3(t,\tau)$ is also finite. In particular,
$\EE[N^2_t]<\infty$ for all $t$.
Since $\Sigma_3(t,\tau)\leqslant \Sigma_1(t+\tau)$,
\begin{equation}
\label{limSigma3}
\lim_{t\to \infty}
\ee^{-\nu(2t+\tau)}\Sigma_3(t,\tau)=0\;.
\end{equation}
Collecting (\ref{Sigma123}), (\ref{limSigma1}), (\ref{limSigma2}),
(\ref{limSigma3}), and using the fact that
$$
\lim_{t\to \infty}\EE[\ee^{-\nu(2t+\tau)}N_t]
=
\lim_{t\to \infty}\EE[\ee^{-\nu(2t+\tau)}N_{t+\tau}]=0,
$$
one gets:
$$
\lim_{t\to \infty}
\EE[\ee^{-\nu(2t+\tau)}N_tN_{t+\tau}]=\frac{C_2}{4}\;.
$$
Hence:
\begin{eqnarray*}
&&\hspace*{-23mm}\lim_{t\to\infty}
\EE[(\ee^{-\nu t}N_t-\ee^{-\nu(t+\tau)}N_{t+\tau})^2]\\
&=&
\lim_{t\to\infty}
\EE[\ee^{-2\nu t}N^2_t-2
\ee^{-\nu(2t+\tau)}N_tN_{t+\tau}+\ee^{-2\nu(t+\tau)}N^2_{t+\tau}]\\
&=&\frac{C_2}{4}-2\frac{C_2}{4}+\frac{C_2}{4}=0\;,
\end{eqnarray*}
hence the result.
\end{proof}
A reinforcement of (\ref{limSigma1}) and (\ref{limSigma2}) ensures
almost sure convergence.
\begin{prop}
\label{prop:expgrowthas}
Under the hypotheses of Proposition \ref{prop:expgrowthL2}, assume
that $W$ is almost surely positive and that:
\begin{equation}
\label{limSigma1as}
\int_0^{\infty} \ee^{-2\nu t}\Sigma_1(t)\,\dd t< \infty\;,
\end{equation}
\begin{equation}
\label{limSigma2as}
\int_0^{\infty}\left|\ee^{-2\nu t} \Sigma_2(t,0) -
C_2\right|\,\dd t < \infty\;.
\end{equation}
Then as $t$ tends to infinity,
$\ee^{-\nu t}N_t$ converges almost surely to $W$.
\end{prop}
\begin{proof}
From the proof of Proposition \ref{prop:expgrowthL2}, the additional
hypothesis yields that:
$$
\int_0^\infty \EE[(\ee^{-\nu t}N_t -W)^2]\,\dd t < \infty\;.
$$
Almost sure convergence is deduced exactly as in the proof of Theorem
21.1 p.~148 of \cite{Harris63}. That $W$ is almost surely positive cannot
be obtained without stronger hypotheses. It will be proved for the BMC
model in section \ref{proofBMC}.
\end{proof}
\begin{rem}
The only change for the $k$-ary tree consists of replacing $2$ by $k$
in the definitions of $\Sigma_1$ and $\Sigma_2$.
\end{rem}

\section{Proof of Theorem \ref{th:expgrowthBMC}}
\label{proofBMC}
As already remarked, symmetry (\ref{eq:symmetric}) and invariance
(\ref{eq:invariant}) imply path- and fork-stationarity. We have also
observed that the solution of $2L(\nu)=1$ is such that:
$$
\nu =\inf\{\gamma>0\,,\;\sum_{n=1}^{\infty}2^n\mathcal{L}_n(\gamma) <\infty\,\}\;.
$$
The main ingredient in the proof consists in applying Lemma
\ref{lem:residue} to $B_\nu(t,u)$ defined by (\ref{Btu}), thanks to
condition $(\mathcal{C}_1)$. This
yields:
\begin{equation}
\label{expobound}
\left|\,\ee^{-\nu t}\sum_{n=1}^\infty 2^{n-1}\PP[S_n\leqslant
  t\,|\,T_0=u]
+ \frac{\alpha(\nu,u)}{4\nu L'(\nu)}\,\right| \leqslant 
\frac{\psi(u)}{2\pi}\,\ee^{-\delta t}\;.
\end{equation}
Recall from (\ref{ENt}) and (\ref{Btu}) that:
$$
\EE[N_t] = 1+\sum_{n=1}^\infty 2^{n-1}\PP[S_n\leqslant t]=
1+A_\nu(t)=1+\int_{\RR^+} B_\nu(t,u)\,\dd\mu(u)\;.
$$
Integrating against $\mu$ (condition \textit{4} of Definition
\ref{def:multergo}), one gets:
$$
\lim_{t\to \infty} \ee^{-\nu t}\EE[N_t] = \int_{\RR^+} C(u)\,\dd\mu(u) = 
- \frac{\alpha(\nu)}{4\nu L'(\nu)} = C\;. 
$$
Consider now
$\displaystyle{\Sigma_1(t)=\sum_{n=0}^{\infty}(n+1)2^n\,\PP[S_n\leqslant t]}$.
The series
$\sum (n+1) 2^n L^n(\gamma)$ converges for $\gamma>\nu$, diverges for
$\gamma\leqslant \nu$. Choose $\gamma>\nu$. By
Markov's inequality:
\begin{eqnarray*}
\Sigma_1(t)&=&\sum_{n=0}^{\infty}(n+1)2^n\PP[S_n\leqslant t]\\
&\leqslant& \ee^{\gamma
  t}\sum_{n=0}^{\infty}(n+1)2^n\mathcal{L}_n(\gamma)\;.
\end{eqnarray*}
Therefore, $\Sigma_1(t)$ is finite for all $t$. Take $\gamma$ such
that $\nu<\gamma<2\nu$.
\begin{eqnarray*}
\ee^{-2\nu t}\Sigma_1(t)&\leqslant& \ee^{(\gamma-2\nu)t}
\sum_{n=0}^{\infty}(n+1)2^n\mathcal{L}_n(\gamma)\;.
\end{eqnarray*}
There exists a constant $K_1$ such that for all $t\geqslant 0$,
$\ee^{-2\nu t} \Sigma_1(t)\leqslant K_1\ee^{(\gamma-2\nu) t}$, hence
(\ref{limSigma1}) and (\ref{limSigma1as}).
\vskip 2mm\noindent
The convergence of $\ee^{-\nu(2t+\tau)}\Sigma_2(t,\tau)$
remains to be proved.
Consider:
$$
\ee^{-\nu(2t+\tau)}\Sigma_2(t,\tau) =
\ee^{-\nu(2t+\tau)}
\sum_{n=0}^{\infty}\sum_{i=1}^{\infty}\sum_{j=1}^{\infty}
2^{n+i+j}\,
\PP[\, S_{n,i}^{(0)}\leqslant t \,,\;S_{n,j}^{(1)}\leqslant
t+\tau\,]\;.
$$
We use the Markov property after conditioning on the event:
$$
B_n := \{\,S_{n-1}=u\,,\; T_{0^n}=x\,,\;
T_{0^{n+1}}=y\,,\; T_{0^n1}=z\,\}\;.
$$
By Definition \ref{def:BMC},
$$
\begin{array}{l}
\PP[\, S_{n,i}^{(0)}\leqslant t \,,\;S_{n,j}^{(1)}\leqslant
t+\tau\,|\,B_n]\\[2ex]
\hspace*{1cm}=
\PP[\,S_i\leqslant t-u-x\,|\,T_0=y]\,
\PP[\,S_j\leqslant t+\tau-u-x\,|\,T_0=z]\;.
\end{array}
$$
Therefore:
$$
\begin{array}{l}
\displaystyle{
\ee^{-\nu(2t+\tau)}
\sum_{i=1}^\infty\sum_{j=1}^{\infty} 2^{i+j}
\PP[\, S_{n,i}^{(0)}\leqslant t \,,\;S_{n,j}^{(1)}\leqslant
t+\tau\,|\,B_n]}\\[2ex]
\hspace*{5mm}\displaystyle{=\ee^{-2\nu(u+x)}
\left(\ee^{-\nu (t-u-x)}\sum_{i=1}^\infty 2^i\,
\PP[\,S_i\leqslant t-u-x\,|\,T_0=y]\right)}\\[2ex]
\hspace*{20mm}\displaystyle{
\left(\ee^{-\nu (t+\tau-u-x)}\sum_{j=1}^\infty 2^j\,
\PP[\,S_j\leqslant t+\tau-u-x\,|\,T_0=y]\right)\;.}
\end{array}
$$
By (\ref{expobound}):
$$
\begin{array}{l}
\displaystyle{
\left|\,\left(
\ee^{-\nu (t-u-x)}\sum_{i=1}^\infty 2^i \,
\PP[\,S_i\leqslant t-u-x\,|\,T_0=y]\right)
+\frac{\alpha(\nu,y)}{2\nu L'(\nu)}\,\right|}\\[2ex]
\hspace*{2cm}
\displaystyle{\leqslant \frac{\psi(y)}{\pi}\ee^{-\delta (t-u-x)}}\;,
\end{array}
$$
and
$$
\begin{array}{l}
\displaystyle{
\left|\,\left(
\ee^{-\nu (t+\tau-u-x)}\sum_{j=1}^\infty 2^j \,
\PP[\,S_j\leqslant t+\tau-u-x\,|\,T_0=z]\right)
+\frac{\alpha(\nu,z)}{2\nu L'(\nu)}\,\right|}\\[2ex]
\hspace*{3cm}
\displaystyle{\leqslant \frac{\psi(z)}{\pi}\ee^{-\delta (t+\tau-u-x)}
\;.}
\end{array}
$$
For $t$ large enough:
$$
\begin{array}{l}
\displaystyle{
\left|\,\left(
\ee^{-\nu (t-u-x)}\sum_{i=1}^\infty 2^i \,
\PP[\,S_i\leqslant t-u-x\,|\,T_0=y]\right)\times
\right.}\\[2ex]
\displaystyle{\left.
\left(
\ee^{-\nu (t\!+\!\tau-u\!-\!x)}\sum_{j=1}^\infty 2^j \,
\PP[\,S_j\leqslant t\!+\!\tau-u\!-\!x|T_0=z]\right)
-\frac{\alpha(\nu,y)\alpha(\nu,z)}{(2\nu L'(\nu))^2}\right|}\\[2ex]
\displaystyle{\leqslant
\frac{-2}{\pi\nu L'(\nu)}\left(\alpha(\nu,z)\psi(y)+\alpha(\nu,y)\psi(z)\right)
\ee^{-\delta (t-u-x)}
\;.}
\end{array}
$$
Denoting by $Q_n$ the joint distribution of $(S_{n-1},T_{0^n})$,
define:
$$
\begin{array}{l}
\displaystyle{
C_2=\frac{1}{(2\nu L'(\nu))^2} \sum_{n=0}^\infty 2^n
\int_{u,x} \ee^{-2\nu(u+x)}}\\[2ex]
\hspace*{2cm}\displaystyle{
\left(\int_{y,z} \alpha(\nu,y)\,\alpha(\nu,z)\,\dd P(x,(y,z))\right)
\,\dd Q_n(u,x)\;.}
\end{array}
$$
By condition \textit{5} of Definition \ref{def:multergo}, there exists
$K_2$ such that for all $x$,
$$
\int_{y,z} \alpha(\nu,y)\,\alpha(\nu,z)\,\dd P(x,(y,z))\leqslant K_2\;.
$$
Hence:
\begin{eqnarray*}
C_2 &\leqslant&\displaystyle{\frac{K_2}{(2\nu L'(\nu))^2} \sum_{n=0}^\infty 2^n
\int_{u,x} \ee^{-2\nu(u+x)}
\,\dd Q_n(u,x)}\\[2ex]
&=&\displaystyle{\frac{K_2}{(2\nu L'(\nu))^2} \sum_{n=0}^\infty 2^n
\EE[\ee^{-2\nu S_{n-1}+T_{0^n}}] 
}\\[2ex]
&=&\displaystyle{\frac{K_2}{(2\nu L'(\nu))^2} \sum_{n=0}^\infty 2^n
\mathcal{L}_n(2\nu) 
<\infty\;.}
\end{eqnarray*}
One gets:
$$
\begin{array}{l}
\displaystyle{\Big|\,\ee^{-\nu(2t+\tau)}\Sigma_2(t,\tau)-C_2\,\Big|
\leqslant
\frac{-2}{\pi\nu L'(\nu)}\ee^{-\delta t}
\sum_{n=0}^\infty 2^n
\int_{u,x} \ee^{-(2\nu-\delta)(u+x)}}\\[2ex]
\hspace*{1cm}\displaystyle{
\left(\int_{y,z} \left(\alpha(\nu,z)\psi(y)+\alpha(\nu,y)\psi(z)\right)
\,\dd P(x,(y,z))\right)
\,\dd Q_n(u,x)}\;.
\end{array}
$$
From condition $(\mathcal{C}_2)$ and symmetry
(\ref{eq:symmetric}), there exists $K_3$ such that for all $x$:
$$
\frac{-2}{\pi\nu L'(\nu)}
\left(\int_{y,z} \left(\alpha(\nu,z)\psi(y)+\alpha(\nu,y)\psi(z)\right)
\,\dd P(x,(y,z))\right) \leqslant K_3\;.
$$
Therefore:
$$
\begin{array}{l}
\Big|\,\ee^{-\nu(2t+\tau)}\Sigma_2(t,\tau)-C_2\,\Big|\\[2ex]
\displaystyle{\leqslant
K_3\, \ee^{-\delta t}
\sum_{n=0}^\infty 2^n
\int_{u,x} \ee^{-(2\nu-\delta)(u+x)}
\,\dd Q_n(u,x)}\\[2ex]
\displaystyle{=
K_3\, \ee^{-\delta t}
\sum_{n=0}^\infty 2^n \EE[\ee^{-(2\nu-\delta)(S_{n-1}+T_{0^n})}]}
\\[2ex]
\displaystyle{=
K_3\, \ee^{-\delta t}
\sum_{n=0}^\infty 2^n \mathcal{L}_n(2\nu-\delta)}\;.
\end{array}
$$
For $\delta<\nu$, the series converges, hence (\ref{limSigma2as}). What
has been proved implies that $\Sigma_2(t,\tau)$ is finite for all $\tau$
and for $t$ large enough. But $\Sigma_2(t,\tau)$ is a nondecreasing
function of $t$, hence it is finite for all $t$ and $\tau$.
\vskip 2mm\noindent
Only one point remains to be proved, that the limit of $\ee^{-\nu t}
N_t$ is almost surely positive.
Assume $T_0=u$ and take $t>u$. Cells alive at time $t$ descend either
from $00$ or from $01$. Therefore:
$$
N_t = N_{t-u}^{(0)} + N_{t-u}^{(1)}\;,
$$
where $(N_s^{(0)})_{s\geqslant 0}$ and $(N_s^{(1)})_{s\geqslant 0}$
have the same distribution as $(N_s)_{s\geqslant 0}$. Multiply by
$\ee^{-\nu t}$:
$$
\ee^{-\nu t} N_t = \ee^{-\nu u}\Big(
\ee^{-\nu(t-u)}N_{t-u}^{(0)}+
\ee^{-\nu(t-u)}N_{t-u}^{(1)}\Big)\;.
$$
Taking the limit in $L^2$ as $t$ tends to infinity, the conditional
distribution of $W$ on $T_0=u$ is the same as the distribution of
$\ee^{-\nu u}(W^{(0)}+W^{(1)})$, where $W^{(0)}$ and $W^{(1)}$ have
the same distribution as $W$. In particular, for all $u>0$,
$$
\PP[W=0\,|\,T_0=u] = \PP[W^{(0)}=0, W^{(1)}=0] \leqslant \PP[W=0]\;.
$$
Hence $\PP[W=0\,|\,T_0=u] = \PP[W=0]$ $\mu$-a.e.
Let
$$p=\PP[W=0]=\PP[W^{(0)}=0, W^{(1)}=0]\;.$$
$$
\begin{array}{l}
p=\\
\displaystyle{\int_{(\RR^+)^3}
\PP[W^{(0)}=0, W^{(1)}=0\,|\, (T_0,T_{00},T_{01})=(x,y,z)]
\dd P(x,(y,z))\dd \mu(x)}\\
\displaystyle{=\int_{(\RR^+)^3}
\PP[W^{(0)}=0\,|\, T_{00}=y]\,
\PP[W^{(1)}=0\,|\, T_{01}=z]
\dd P(x,(y,z))\dd \mu(x)\;,}
\end{array}
$$
by Definition \ref{def:BMC}. Since
$\PP[W^{(0)}=0\,|\, T_{00}=y]=\PP[W^{(1)}=0\,|\, T_{01}=z]=p$, it
follows that $p=p^2$. But $p=1$ is excluded since $\EE[W]>0$. Hence $p=0$.
\section*{Acknowledgements}
The authors which to thank Alain Le Breton for the explicit
calculations of section \ref{explicit}, and Christine Laurent-Thi\'ebaut for
her expertise on Tauberian lemmas.
\bibliographystyle{imsart-number}

%
%
\end{document}